\documentclass[conference]{IEEEtran}
\IEEEoverridecommandlockouts

\usepackage[utf8]{inputenc} 
\usepackage[T1]{fontenc}

\usepackage{booktabs}
\usepackage{tabularx}
\usepackage{array}
\usepackage{makecell}
\usepackage{threeparttable}
\usepackage[table]{xcolor}
\usepackage{pifont}
\usepackage{multirow}

\newcommand{\cmark}{\ding{51}}
\newcommand{\xmark}{\ding{55}}

\newcommand{\pmark}{\ensuremath{\triangle}}
\definecolor{myblue}{RGB}{217,231,245}
\definecolor{myblue2}{RGB}{235,243,250}
\definecolor{mygray}{RGB}{245,245,245}
\definecolor{ourblue}{RGB}{210,226,242}
\definecolor{HeaderBlue}{HTML}{17365D}
\definecolor{LightBlue}{HTML}{EEF5FB}
\definecolor{LightGreen}{HTML}{EDF7ED}
\definecolor{LightTeal}{HTML}{EAF7F7}
\definecolor{LightGray}{HTML}{F5F5F5}

\definecolor{PriorGray}{RGB}{247,247,247}
\definecolor{OursBlueA}{RGB}{226,239,252}
\definecolor{OursBlueB}{RGB}{207,226,245}
\definecolor{SoftRule}{RGB}{190,190,190}
\usepackage{bm}
\usepackage{amssymb, amsthm, amsfonts, mathtools, marvosym}
\newtheorem{definition}{Definition}
\newtheorem{theorem}{Theorem}
\newtheorem{lemma}{Lemma}
\newtheorem{corollary}{Corollary}

\newtheorem{assumption}{Assumption}
\newtheorem{remark}{Remark}

\theoremstyle{plain}
\newtheorem{example}{Example}

\newcommand{\X}{\mathcal{X}}
\newcommand{\Y}{\mathcal{Y}}

\newcommand{\R}{\mathbb{R}}
\newcommand{\E}{\mathbb{E}}
\newcommand{\supp}{\mathrm{supp}}

\newcommand{\pc}{\mathcal{P}}
\newcommand{\D}{\mathrm{d}}

\newcommand{\dist}{\mathrm{dist}}

\makeatletter
\newsavebox\myboxA
\newsavebox\myboxB
\newlength\mylenA
\newcommand*\xoverline[2][0.75]{%
\sbox{\myboxA}{$\m@th#2$}%
\setbox\myboxB\null
\ht\myboxB=\ht\myboxA%
\dp\myboxB=\dp\myboxA%
\wd\myboxB=#1\wd\myboxA
\sbox\myboxB{$\m@th\overline{\copy\myboxB}$}
\setlength\mylenA{\the\wd\myboxA}
\addtolength\mylenA{-\the\wd\myboxB}%
\ifdim\wd\myboxB<\wd\myboxA%
\rlap{\hskip 0.5\mylenA\usebox\myboxB}{\usebox\myboxA}%
\else
\hskip -0.5\mylenA\rlap{\usebox\myboxA}{\hskip 0.5\mylenA\usebox\myboxB}%
\fi}
\makeatother
\usepackage{graphicx}
\usepackage{cite}
\usepackage{hyperref} 
\usepackage{color}
\usepackage{xcolor}
\usepackage{subcaption}
\usepackage{extarrows}
\usepackage{url}
\usepackage{ifthen}

\def\BibTeX{{\rm B\kern-.05em{\sc i\kern-.025em b}\kern-.08em
    T\kern-.1667em\lower.7ex\hbox{E}\kern-.125emX}}
\begin{document}

\title{Rate-distortion Theory with Lower Semi-continuous Distortion
on Noncompact Alphabets
}

\author{
\IEEEauthorblockN{Jiayang Zou\textsuperscript{1,2}, Luyao Fan\textsuperscript{2}, Jiayang Gao\textsuperscript{2}, and Jia Wang\textsuperscript{2}}
\IEEEauthorblockA{\textsuperscript{1}Stanford University, Stanford, CA, USA; jyangzou@stanford.edu}
\IEEEauthorblockA{\textsuperscript{2}Shanghai Jiao Tong University, Shanghai, China; \{qiudao, fanluyao, gjy0515, jiawang\}@sjtu.edu.cn}

}

\maketitle

\begin{abstract}
    In this paper, we study rate-distortion theory for general sources with an emphasis on the existence of optimal reconstruction distributions on noncompact alphabets. Classical attainability results typically rely on compactness of the reproduction alphabet together with continuity of the distortion function, which may fail in many noncompact settings. We identify two complementary existence mechanisms under lower semi-continuity on locally compact Polish alphabets. For bounded distortions, we prove that the rate-distortion infimum is attained via the one-point compactification argument. For unbounded coercive distortions, we establish existence via concentration-compactness. We also give several counterexamples showing that our attainability results are close to sharp. Our results provide a unified and transparent existence theorem for rate-distortion problems with lower semi-continuous distortions.
\end{abstract}

\section{Introduction}
\label{sec: RD_modern_intro}

Rate-distortion theory is a central branch of information theory, concerned with characterizing the fundamental limits of source compression when a prescribed level of distortion is permitted. 
Given a source random variable $X$ with probability distribution $p(x)$ and a distortion (or loss) function $\rho(x,y)$, the rate-distortion function $R(D)$ is defined as the minimum information rate required to represent the source subject to an average distortion constraint $D$:
\begin{equation}
    R(D)=\inf_{p(y|x)} I(X;Y),
\end{equation}
where the infimum is taken over all conditional distributions (or disintegrations) satisfying the fidelity constraint
\begin{equation}
    \E[\rho(X,Y)]\leqslant D.
\end{equation}
Classical formulations of the rate-distortion problem primarily focus on finite-alphabet sources or continuous sources defined on Euclidean spaces. Comprehensive treatments can be found in the monographs of Berger \cite{berger2003rate}, Blahut \cite{blahut1972computation,blahut1987principles}, and Cover and Thomas \cite{cover2006elements}.

To place rate-distortion theory on a rigorous mathematical footing, Csisz\'ar established in his seminal work \cite{csiszar1974extremum} a collection of fundamental properties of $R(D)$ and its parametric representations under remarkably weak assumptions. 
These results provide a foundational framework for source coding problems on general abstract spaces.
Further contributions were made by E. Riegler \emph{et al.} \cite{riegler2018rate, riegler2023lossy}, who focused on lossy compression theory for sources defined on compact manifolds and fractal sets. By leveraging results related to the rate-distortion dimension \cite{kawabata1994rate}, they derived the Shannon lower bound (SLB) in more general settings.
In terms of the parametric representation of the RD function, V. Kostina and E. Tuncel \cite{kostina2019successive} streamlined I. Csisz\'ar's argument by utilizing the Donsker-Varadhan characterization of minimum relative entropy, which generalized the earlier results of Equitz and Cover \cite{equitz1991successive}. More recently, connections between rate-distortion theory and optimal weak transport have attracted increasing attention \cite{gozlan2017kantorovich,zou2025revisitratedistortionproblemsoptimal}.

For existence of optimal reconstructions in general settings, classical compactness-based arguments \cite{csiszar1974extremum} relying on Prokhorov's theorem no longer apply directly once the reproduction alphabet is noncompact. As a consequence, minimizing sequences of the rate-distortion functional may fail to admit convergent subsequences, exhibiting pathological behaviors such as mass escaping to infinity or splitting into widely separated components. A significant step toward relaxing compactness assumptions was made by Rezaei \emph{et al.} \cite{rezaei2006rate}, who developed a weak-* compactified existence theory under continuity assumptions on the distortion. In the noncompact case, their formulation effectively passes to a relaxed compactified problem; this, however, does not by itself imply the existence of an optimal reconstruction on the original reproduction alphabet.

It is also worth noting that there are important specialized existence results in Euclidean settings. In particular, for finite-dimensional Euclidean alphabets and distortions of the form
$d(x,y)=\|x-y\|^r,(r>0)$,
Mao \emph{et al.} \cite{mao2021rateconstrained} recorded that an achieving distribution exists whenever $\E\|X\|^r<\infty$. This covers, for example, the classical squared-error distortion on $\mathbb R^n$. However, such arguments are tailored to Euclidean norm-power distortions and do not provide a general original-space existence theory for noncompact abstract alphabets. 

In this paper, we separate two complementary regimes. For bounded distortions on a noncompact locally compact Polish alphabet, we show that a one-point compactification together with a compactification-dominance condition yields attainability for the rate-distortion problem. This covers the Hamming distortion or the $0$-$1$ loss. For unbounded coercive distortions, we return to the concentration-compactness principle of Lions \cite{lions1984concentration1,lions1984concentration2,lions1985concentration1,lions1985concentration2}. By combining concentration-compactness with a compatible Heine-Borel metric on every locally compact Polish alphabet, we establish existence of the optimal reconstruction. This covers the classical squared-error distortion on $\mathbb R^n$ and other unbounded coercive distortions, such as the Itakura-Saito distortion \cite{Watanabe2020discrete}, the hard-thresholding distortion in JND settings. We also give several counterexamples showing that our results are close to sharp.
Table~\ref{tab:existence-comparison} summarizes how these two regimes compare with the main earlier existence results.

\begin{table*}[t]
\centering
\caption{Comparison with related existence results for optimal reconstructions. The first column lists the relevant prior works and the two theorem regimes proved in this paper. ``Noncompact reproduction alphabet'' asks whether the reproduction space may be noncompact. ``Original-space optimizer'' asks whether the result yields an optimizer on the original reproduction alphabet. ``l.s.c. distortion'' asks whether lower semi-continuous distortions are covered. ``\(0\)-\(1\) / Hamming'' asks whether exact \(0\)-\(1\) or countable Hamming distortions are included. ``General non-Euclidean coercive losses'' asks whether coercive losses beyond Euclidean norm-power models are treated. ``Tightness of minimizing laws'' asks whether tightness or compactness of minimizing output laws is proved. We write \cmark for yes, \xmark for no or not addressed, and \pmark for partial coverage or special cases. The row ``Specialized Euclidean norm-power results'' refers to finite-dimensional Euclidean models such as squared error; see, for example, \cite{mao2021rateconstrained}. The row for Rezaei et al.~\cite{rezaei2006rate} concerns a weak-\(*\)/compactified relaxed formulation, so it does not by itself yield an ordinary optimizer on the original reproduction alphabet.}
\label{tab:existence-comparison}
\small
\setlength{\aboverulesep}{0pt}
\setlength{\belowrulesep}{0pt}
\setlength{\extrarowheight}{0.3ex}
\setlength{\tabcolsep}{4pt}
\renewcommand{\arraystretch}{1.08}
\begin{tabular}{
>{\raggedright\arraybackslash}m{3.00cm}
>{\centering\arraybackslash}m{1.75cm}
>{\centering\arraybackslash}m{2.05cm}
>{\centering\arraybackslash}m{1.90cm}
>{\centering\arraybackslash}m{1.50cm}
>{\centering\arraybackslash}m{2.00cm}
>{\centering\arraybackslash}m{1.90cm}
}
\toprule
\makecell[c]{\textbf{Result / work}}
&
\makecell[c]{\textbf{Noncompact}\\ \textbf{reproduction}\\ \textbf{alphabet}}
&
\makecell[c]{\textbf{Original-space}\\ \textbf{optimizer}}
&
\makecell[c]{\textbf{l.s.c. }\\ \textbf{distortion}}
&
\makecell[c]{\textbf{\(0\)-\(1\) /}\\ \textbf{Hamming}}
&
\makecell[c]{\textbf{General}\\ \textbf{coercive}\\ \textbf{losses}}
&
\makecell[c]{\textbf{Tightness}\\ \textbf{of minimizing}\\ \textbf{laws}}
\\
\midrule
\multicolumn{7}{l}{\textbf{Prior literature}}\\
\midrule

\rowcolor{PriorGray}
Csisz\'ar~\cite{csiszar1974extremum}
& \xmark & \cmark & \xmark & \xmark & \xmark & \cmark
\\

\rowcolor{PriorGray}
Rezaei et al.~\cite{rezaei2006rate}
& \cmark & \xmark & \xmark & \xmark & \cmark & \xmark
\\

\rowcolor{PriorGray}
\makecell[l]{Specialized Euclidean\\ norm-power results}
& \cmark & \cmark & \xmark & \xmark & \xmark & \pmark
\\

\midrule
\multicolumn{7}{l}{\textbf{This paper}}\\
\midrule

\rowcolor{OursBlueB}
\makecell[l]{\textbf{Bounded regime}\\ Theorem~\ref{thm:bounded-itw}}
& \cmark & \cmark & \cmark & \cmark & \xmark & \xmark
\\

\rowcolor{OursBlueB}
\makecell[l]{\textbf{Unbounded regime}\\ Theorem~\ref{thm:coercive-itw}}
& \cmark & \cmark & \cmark & \xmark & \cmark & \cmark
\\

\bottomrule
\end{tabular}
\end{table*}

The remainder of this paper is organized as follows.
Section~\ref{sec: RD_recap} provides the problem setting with lower semi-continuous distortion.
In Section~\ref{sec: main_results}, we present two complementary existence theorems: a bounded compactification theorem and a coercive concentration-compactness theorem on locally compact Polish alphabets with several counterexamples to demonstrate the sharpness of our results. In Section~\ref{sec: Discussion}, we conclude with a discussion of
potential extensions and applications of our framework.

\section{The Problem Setting}
\label{sec: RD_recap}

We begin by imposing some basic regularity conditions on the distortion function, following \cite{csiszar1974extremum}.

\begin{assumption}\label{assumption: distortion function basic requirement}
Let $\rho:\X\times\Y\to [0,\infty]$ be a measurable function, referred to as the loss function. Assume that
\begin{equation}
    \inf_{y\in \Y} \rho(x,y)=0, \quad \forall x\in\X.
\end{equation}
\end{assumption}

\begin{assumption}\label{assumption: existence of finite set}
Assume that there exists a finite set $B\subset \Y$ such that
\[
\int \rho(x,B)\,\mu(\mathrm{d}x)<\infty,
\qquad\text{where}\qquad
\rho(x,B)\coloneqq \min_{y\in B}\rho(x,y).
\]
Then one may choose a finite-valued reconstruction $Y=y(X)$ satisfying
\[
\rho(X,Y)=\rho(X,B)\quad\text{a.s.},
\]
and consequently
\[
\E \rho(X,Y)<\infty,
\qquad
I(X;Y)\le H(Y)<\infty.
\]
Thus Assumption~\ref{assumption: existence of finite set} guarantees finite-information feasibility.
\end{assumption}

Based on these assumptions, Gray provided a systematic reformulation of Csisz\'ar's theory, recasting $R(D)$ into a compact variational form.

\begin{theorem}{}[Theorem 9.2 in \cite{gray2011entropy}]
\label{thm: parametric_representation_Gray}
If $R(D)<\infty$, then the rate-distortion function admits the representation
\begin{equation}
\begin{aligned}
R(D)&=\max_{\beta\geqslant 0}\left\{F(\beta)-\beta D\right\},\\
F(\beta)&=\inf_{\nu\in\mathcal{P}(\Y)}
\int \log \frac{1}{\int e^{-\beta\rho(x,y)}\,\mathrm{d}\nu}\,\mathrm{d}\mu,
\end{aligned}
\end{equation}
where the infimum is taken over all probability measures $\nu$ on the reproduction space $\Y$.
This formulation reduces the computation of $R(D)$ to a coupled optimization over the output distribution $\nu$ and the Lagrange multiplier $\beta$.
\end{theorem}

In classical existence theorem of optimal reconstructions, compactness of the reproduction space plays a decisive role. Once the reproduction alphabet becomes noncompact, minimizing sequences of $F(\beta)$ in Theorem~\ref{thm: parametric_representation_Gray} may not admit convergent subsequences in the original space (the term "original" refers to the original $\Y$ compared to the compactified space in \cite{rezaei2006rate}). We study original-space attainability under lower semi-continuity on noncompact alphabets through two complementary regimes: bounded distortions handled by compactification and unbounded coercive distortions handled by concentration-compactness.

The following examples motivate this distinction and illustrate distortions beyond the reach of compact-alphabet arguments or specialized Euclidean norm-power existence results.


\begin{itemize}

\item \textbf{Exact $0$-$1$ loss / countable Hamming distortion.}
Let $\X=\Y$ be a locally compact Polish space and define
\[
\rho_0(x,y)=\mathbf{1}_{\{x\neq y\}}.
\]
Since Polish spaces are Hausdorff, the diagonal
\[
\Delta=\{(x,x):x\in\X\}
\]
is closed in $\X\times\X$. Hence $\{(x,y):x\neq y\}$ is open, and $\rho_0$ is bounded and lower semi-continuous. 

\item \textbf{Perceptual coding with a dead-zone threshold.}
In perceptual image and video compression, the concept of \emph{just noticeable difference} (JND) suggests that reconstruction errors below a visibility threshold $\tau$ should incur no penalty \cite{donoho1995noising}. For instance, one may consider
\[
    \rho_{\mathrm{DZ}}(x,y)
    =
    \|x-y\|^2\,\mathbf{1}_{\{\|x-y\|>\tau\}}.
\]
This function is lower semi-continuous but discontinuous at the threshold boundary $\|x-y\|=\tau$. It is unbounded and coercive on $\mathbb R^n$, and therefore falls under the concentration-compactness regime.

\item \textbf{Itakura--Saito distortion.}
For positive-valued sources, the scalar Itakura--Saito distortion is defined on $\X=\Y=(0,\infty)$ by
\[
\rho_{\mathrm{IS}}(x,y)
=x/y
-
\log(x/y)
-
1,
\qquad x,y>0.
\]
This distortion is topologically coercive and continuous, but it is not a Euclidean norm-power distortion. It is widely used in audio and speech processing \cite{Watanabe2020discrete}.

\item \textbf{Distortion with penalty jumps.}
In robust estimation and discontinuous optimization, one often imposes a large penalty once the reconstruction error exceeds a prescribed tolerance \cite{conn1998discontinuous, mongeau2008discontinuous, imo1984discontinuous, fourer1985simplex}. A representative example is
\[
\rho(x, y) = 
\begin{cases} 
\|x - y\|^2, & \text{if } \|x - y\| \le \tau_1, \\[0.3em]
\|x - y\|^2 + J, & \text{if } \|x - y\|>\tau_1,
\end{cases}
\]
where $J>0$ is a large constant. It is lower semi-continuous, and its compact sublevel sets place it in the coercive regime.

\end{itemize}

\section{Main Results}
\label{sec: main_results}

We now distinguish two mechanisms guaranteeing original-space attainability on a noncompact reproduction alphabet.
\subsection{Bounded lower semi-continuous distortions}

Assume throughout this subsection that $\Y$ is a noncompact locally compact Polish space. 
Since $\Y$ is a locally compact Hausdorff space and hence, in particular, a Tychonoff space, 
we write
\[
\Y^\ast:=\Y\cup\{\infty\}
\]
for the Alexandroff one-point compactification of $\Y$. 
The topology on $\Y^\ast$ is defined by retaining the original open sets of $\Y$ and by taking the sets
\[
\{\infty\}\cup(\Y\setminus K),
\qquad K\subset\Y\ \text{compact},
\]
as a neighborhood basis at the added point $\infty$. 
With this topology, $\Y^\ast$ is compact Hausdorff; since $\Y$ is locally compact Polish, $\Y^\ast$ is also compact metrizable. 
The symbol $\infty$ denotes only the compactification boundary point and does not presuppose any Euclidean or order structure on $\Y$.

\begin{definition}[Compactified extension and compactification-dominance]
\label{def:CD-itw}
Let $\rho:\X\times\Y\to[0,M]$ be a bounded distortion. 
We say that $\rho$ admits a bounded lower semi-continuous compactified extension if there exists a Borel measurable function
\[
\rho^\ast:\X\times\Y^\ast\to[0,M]
\]
such that
\[
\rho^\ast(x,y)=\rho(x,y),
\qquad (x,y)\in\X\times\Y,
\]
and for $\mu$-almost every $x\in\X$, the section
\[
y\mapsto \rho^\ast(x,y)
\]
is lower semi-continuous on $\Y^\ast$.

Given such an extension $\rho^\ast$, we say that the pair $(\rho,\rho^\ast)$ satisfies the \emph{compactification-dominance condition} if there exists $y_0\in\Y$ such that
\begin{equation}
\rho(x,y_0)\le \rho^\ast(x,\infty),
\qquad \mu\text{-a.e. }x.
\label{eq:CD-itw}
\end{equation}
\end{definition}

\begin{remark}
In general one cannot set $\rho^\ast(x,\infty)=\sup_{y\in\Y}\rho(x,y)$, since lower semi-continuity at $\infty$ would require $\sup_{y\in\Y}\rho(x,y)\le\liminf_{n\to\infty}\rho(x,y_n)$ for every sequence $y_n\to\infty$, which may fail. The compactification-dominance condition means that any boundary mass at $\infty$ can be pushed back to some $y_0\in\Y$ without increasing the objective.
\end{remark}

\begin{theorem}[]
\label{thm:bounded-itw}
Assume that:
\begin{enumerate}
    \item $\X$ is Polish, $\mu\in\pc(\X)$, and $\Y$ is a noncompact locally compact Polish space;
    \item $\rho:\X\times\Y\to[0,M]$ admits a bounded lower semi-continuous compactified extension
    \[
    \rho^\ast:\X\times\Y^\ast\to[0,M]
    \]
    in the sense of Definition~\ref{def:CD-itw};
    \item the pair $(\rho,\rho^\ast)$ satisfies the compactification-dominance condition \eqref{eq:CD-itw}.
\end{enumerate}
Then, for every $\beta\ge0$, the output-law functional
\[
J_\beta(\nu):=-\int_\X\log\!\left(\int_\Y e^{-\beta\rho(x,y)}\,\nu(\D y)\right)\mu(\D x)
\]
attains its infimum on $\pc(\Y)$.
\end{theorem}

\begin{proof}[Proof sketch]
See Appendix~\ref{app:bounded-itw} for details. The case $\beta=0$ is trivial, since then $J_0(\nu)=0$ for every $\nu\in\pc(\Y)$. 
Assume $\beta>0$. 
First minimize the compactified functional
\[
J_\beta^\ast(\nu):=
-\int_\X
\log\!\left(
\int_{\Y^\ast} e^{-\beta\rho^\ast(x,y)}\,\nu(\D y)
\right)\mu(\D x),\; \nu\in\pc(\Y^\ast).
\]
Since $\Y^\ast$ is compact, $\pc(\Y^\ast)$ is weakly compact.

We claim that $J_\beta^\ast$ is weakly lower semi-continuous. 
Let $\nu_n\rightharpoonup\nu$ in $\pc(\Y^\ast)$. For $\mu$-a.e. $x$, the function
\[
y\mapsto \rho^\ast(x,y)
\]
is lower semi-continuous, hence
\[
y\mapsto e^{-\beta\rho^\ast(x,y)}
\]
is bounded upper semi-continuous. Therefore, by the Portmanteau theorem,
\[
\int_{\Y^\ast}e^{-\beta\rho^\ast(x,y)}\,\nu(\D y)
\ge
\limsup_{n\to\infty}
\int_{\Y^\ast}e^{-\beta\rho^\ast(x,y)}\,\nu_n(\D y).
\]
Because $z\mapsto-\log z$ is decreasing and continuous on $[e^{-\beta M},1]$, it follows that
\begin{equation}
    \begin{aligned}
    -\log\!&\left(
\int_{\Y^\ast}e^{-\beta\rho^\ast(x,y)}\,\nu(\D y)
\right)\\
&\le \liminf_{n\to\infty}
-\log\!\left(
\int_{\Y^\ast}e^{-\beta\rho^\ast(x,y)}\,\nu_n(\D y)
\right)
    \end{aligned}
\end{equation}
for $\mu$-a.e. $x$. Fatou's lemma gives
\[
J_\beta^\ast(\nu)
\le
\liminf_{n\to\infty}J_\beta^\ast(\nu_n).
\]
Thus $J_\beta^\ast$ attains a minimizer $\nu^\ast\in\pc(\Y^\ast)$.

Let $y_0$ be given by the compactification-dominance condition, and define
\[
T:\Y^\ast\to\Y,
\qquad
T(y)=y\ \text{for }y\in\Y,
\qquad
T(\infty)=y_0.
\]
Set
\[
\widehat\nu:=T_\#\nu^\ast\in\pc(\Y).
\]
Since $\rho^\ast$ extends $\rho$ on $\X\times\Y$ and \eqref{eq:CD-itw} holds at the boundary point, we have
\[
\rho(x,T(y))\le \rho^\ast(x,y),
\qquad \mu\text{-a.e. }x,\quad \nu^\ast\text{-a.e. }y.
\]
Hence
\[
e^{-\beta\rho(x,T(y))}
\ge
e^{-\beta\rho^\ast(x,y)}.
\]
Therefore, for $\mu$-a.e. $x$,
\begin{equation}
    \begin{aligned}
        \int_{\Y}e^{-\beta\rho(x,z)}\,\widehat\nu(\D z)
&=
\int_{\Y^\ast}e^{-\beta\rho(x,T(y))}\,\nu^\ast(\D y)\\
&\ge \int_{\Y^\ast}e^{-\beta\rho^\ast(x,y)}\,\nu^\ast(\D y).
    \end{aligned}
\end{equation}
Taking $-\log$ and integrating over $\mu$ yields
\[
J_\beta(\widehat\nu)\le J_\beta^\ast(\nu^\ast).
\]
Since every $\nu\in\pc(\Y)$ is also a probability measure on $\Y^\ast$ with no mass at $\infty$, one has
\[
\inf_{\pc(\Y^\ast)}J_\beta^\ast
\le
\inf_{\pc(\Y)}J_\beta.
\]
The previous inequality gives the reverse inequality through the pushforward $\widehat\nu$. Hence
\[
J_\beta(\widehat\nu)=\inf_{\nu\in\pc(\Y)}J_\beta(\nu),
\]
and the infimum is attained on the original alphabet $\Y$.
\end{proof}


\begin{corollary}[Exact $0$-$1$ loss]
\label{cor:zeroone-itw}
Let $\X=\Y$ be a noncompact locally compact Polish space and define
\[
\rho_0(x,y)=\mathbf{1}_{\{x\neq y\}}.
\]
Set
\[
\rho_0^\ast(x,y)=
\begin{cases}
\mathbf{1}_{\{x\neq y\}}, & y\in\Y,\\[0.3em]
1, & y=\infty.
\end{cases}
\]
Then $\rho_0^\ast$ is bounded, Borel measurable, and for every fixed $x\in\X$ the map
\[
y\mapsto \rho_0^\ast(x,y)
\]
is lower semi-continuous on $\Y^\ast$. Moreover, \eqref{eq:CD-itw} holds for every $y_0\in\Y$. Consequently, the existence conclusion of Theorem~\ref{thm:bounded-itw} holds. 
\end{corollary}

\begin{proof}[Proof sketch]
See Appendix~\ref{app:zeroone-itw} for details. For fixed $x\in\X$, the singleton $\{x\}$ is closed because Polish spaces are Hausdorff. Hence
\(
\Y\setminus\{x\}
\)
is open, and $\mathbf{1}_{\{x\neq y\}}$ is lower semi-continuous on $\Y$. It remains to check lower semi-continuity at the compactification point. If $y_n\to\infty$ in $\Y^\ast$, then $y_n$ eventually leaves every compact subset of $\Y$. Since $\{x\}$ is compact, $y_n\neq x$ eventually. Hence
\(
\rho_0^\ast(x,y_n)\to 1=\rho_0^\ast(x,\infty).
\)
The domination \eqref{eq:CD-itw} is immediate because $\rho_0(x,y_0)\le1$.
\end{proof}
\begin{remark}
    A metric threshold loss $\mathbf{1}_{\{d(x,y)>\tau\}}$ can be handled in the same way whenever it admits a bounded compactified extension whose sections in $y$ are lower semi-continuous. This extension is automatic, for instance, when $d$ has the Heine--Borel property. Indeed, for fixed $x$, if $y_n\to\infty$ in $\Y^\ast$, then $y_n$ eventually leaves the compact ball $\{y:d(x,y)\le\tau\}$, and hence $d(x,y_n)>\tau$ eventually.
\end{remark}

\subsection{Coercive lower semi-continuous distortions}

For the unbounded regime we use concentration-compactness as an intrinsic compactness mechanism on the original reproduction alphabet. 
The topological assumption ``locally compact Polish'' is sufficient for this purpose, since such spaces admit compatible metrics with the Heine-Borel property. 
Recall that a metric has the Heine-Borel property if every closed and bounded set is compact. Hence we assume throughout this subsection that $\Y$ is a noncompact locally compact Polish space equipped with a compatible metric $d$ satisfying the Heine-Borel property.

\begin{lemma}[Compatible Heine--Borel remetrization]
\label{lem:HB-remetrization}
Every noncompact locally compact Polish space admits a metric $d_{\mathrm{HB}}$ that generates the original topology and satisfies the Heine--Borel property.
\end{lemma}
This is a standard remetrization fact; see, for example, \cite{williamson1987constructing}. 


\begin{definition}[Inf-compactness in the reproduction variable]
\label{def:infcompact-y}
We say that $\rho:\X\times\Y\to[0,\infty]$ is \emph{inf-compact in the reproduction variable} if, for $\mu$-a.e. $x\in\X$ and every $L<\infty$, the sublevel set
\[
K_{x,L}:=\{y\in\Y:\rho(x,y)\le L\}
\]
is compact in $\Y$.
\end{definition}

\begin{remark}
The terminology ``inf-compact'' \cite{feinberg2012average} is standard in analysis and optimization: it means that all lower sublevel sets are compact. In a Hausdorff space, compact sublevel sets are closed; hence inf-compactness implies lower semi-continuity in the corresponding variable. In a metric space with Heine-Borel property, inf-compactness is equivalent to lower semi-continuity plus the coercivity ($d(y,y_0)\to\infty
\Longrightarrow
\rho(x,y)\to+\infty$ for $\mu$-a.e. $x\in\X$). Hence we use the term ``coercive'' in the main text.
\end{remark}

We next recall the concentration-compactness alternative in the form needed below, which is a straightforward extension of the classical result of Lions \cite{lions1984concentration1} to general Polish spaces. 
The point of Lemma~\ref{lem:HB-remetrization} is that, on a locally compact Polish alphabet, we may choose a compatible metric whose closed bounded sets are compact. 
Thus the metric concentration alternatives below can be interpreted as alternatives for loss of compactness in the original topology of $\Y$.

\begin{theorem}[Generalized concentration-compactness lemma I]
\label{thm:gcc-itw}
Let $\{\nu_n\}$ be a sequence of probability measures on an unbounded Polish space $(\Y,d)$. Then there exists a subsequence, still denoted by $\{\nu_n\}$, such that one of the following alternatives holds:
\begin{enumerate}
    \item \emph{Compactness.} For every $\varepsilon>0$, there exist $R<\infty$ and centers $y_n\in\Y$ such that
    \[
    \nu_n(B_d(y_n,R))\ge 1-\varepsilon
    \qquad\text{for all }n.
    \]

    \item \emph{Vanishing.} For every $R>0$,
    \[
    \sup_{y\in\Y}\nu_n(B_d(y,R))\to0.
    \]

    \item \emph{Dichotomy.} There exists $\alpha\in(0,1)$ such that for every $\varepsilon>0$, there exist $R>0$ and centers $y_n\in\Y$ with the following property: for every $R'>R$, there exist nonnegative measures $\nu_{n}^1,\nu_{n}^2$ satisfying
    \[
    0\le \nu_{n}^1+\nu_{n}^2\le \nu_n,
    \]
    \[
    \supp (\nu_{n}^1)\subset B_d(y_n,R),
    \qquad
    \supp (\nu_{n}^2)\subset B_d(y_n,R')^c,
    \]
    \[
    \limsup_{n\to\infty}
    \left(
    \left|\alpha-\nu_{n}^1(\Y)\right|
    +
    \left|(1-\alpha)-\nu_{n}^2(\Y)\right|
    \right)
    \le \varepsilon,
    \]
    and
    \[
    \operatorname{dist}
    \bigl(
    \supp (\nu_{n}^1),
    \supp (\nu_{n}^2)
    \bigr)
    \to\infty.
    \]
\end{enumerate}
Here
\(
\operatorname{dist}(A,B):=\inf\{d(a,b):a\in A,\ b\in B\}.
\)
\end{theorem}
\begin{proof}
    See Appendix~\ref{app:gcc-itw} for details.
\end{proof}

The concentration-compactness argument is used to analyze how a minimizing sequence can fail to be tight in the original reproduction alphabet. Vanishing is ruled out because it would make the exponential integral inside the logarithm collapse for almost every source symbol, forcing the objective value to diverge rather than approach the optimum. Dichotomy is ruled out by using the separation of supports in an essential way. Once two principal components are forced far apart, at least one of them must escape every compact subset of the reproduction space. Coercivity then makes the escaping component asymptotically negligible. A small mass-transfer argument shows that such a split cannot be minimizing. Hence neither vanishing nor dichotomy can occur, and minimizing sequences must be tight in the original alphabet.

\begin{theorem}
\label{thm:coercive-itw}
Assume that $\X$ is Polish, $\mu\in\pc(\X)$, and $\Y$ is a noncompact locally compact Polish space equipped with a compatible Heine--Borel metric. Suppose that $\rho$ is lower semi-continuous and coercive in the reproduction variable in the sense of Definition~\ref{def:infcompact-y}. Fix $\beta>0$ and define
\[
m_\beta:=\inf_{\nu\in\pc(\Y)}J_\beta(\nu).
\]
Assume $m_\beta<\infty$. Then every minimizing sequence
\[
J_\beta(\nu_n)\to m_\beta
\]
is tight in $\pc(\Y)$. Consequently, $J_\beta$ attains its infimum on $\pc(\Y)$, and every weak cluster point of a minimizing sequence is a minimizer.
\end{theorem}
\begin{proof}
    The details are given in the Appendix \ref{app:coercive-itw}.
\end{proof}
\begin{remark}
The coercive case can also be viewed through the one-point compactification by setting $\rho^\ast(x,\infty)=+\infty$. Such an argument can give a short proof of bare existence for the output-law functional. The concentration-compactness argument gives a stronger statement: it proves compactness of minimizing sequences in the original alphabet. More precisely, it shows that mass cannot vanish or split into widely separated components, and therefore any minimizing sequence is tight. This is a stability statement for approximation procedures.
\end{remark}
This theorem covers, for instance, the classical squared-error distortion on $\mathbb R^n$, the Itakura--Saito distortion on the positive orthant, and the lower semi-continuous dead-zone and penalty-jump distortions introduced in Section~\ref{sec: RD_recap}.
It is useful, however, to emphasize that the assumptions in Theorem~\ref{thm:coercive-itw} are close to sharp. The coercivity, local-compactness and lower semi-continuity assumptions cannot be omitted directly. The details of the following examples are given in the Appendix~\ref{app:counterexamples}, where we provide a uniform framework for constructing such counterexamples to existence of minimizers under various relaxations of the assumptions in Theorem~\ref{thm:coercive-itw}.


\begin{example}[Coercivity is essential]
\label{prop:counter-itw}
Let $\X=\{\ast\}$, $\Y=\R$, and
\[
\rho(\ast,y)=\frac{1}{1+y^2}.
\]
Then for every $\beta>0$,
\[
\inf_{\nu\in\pc(\R)}J_\beta(\nu)=0,
\]
but the infimum is not attained by any $\nu\in\pc(\R)$.
\end{example}

\begin{example}[Local compactness is essential]
\label{ex:nonlocally-compact-counter}
Let $\X=\{0,1\}$ with $\mu(0)=\mu(1)=1/2$, and let
\(
\Y=\ell^2
\)
with its usual Hilbert norm. This is a Polish space but is not locally compact. Let $\{e_n\}_{n\ge1}$ denote the standard orthonormal basis of $\ell^2$. Define
\[
r_+(y)\!:=\!\inf_{n\ge1}\!\left(\!\frac1n+\|y-e_n\|\!\right)\!,
r_-(y)\!:=\!\inf_{n\ge1}\!\left(\!\frac1n+\|y+e_n\|\!\right).
\]
The functions $r_+$ and $r_-$ are $1$-Lipschitz and are strictly positive at every fixed $y\in\ell^2$. Set
\[
q_1(y):=\frac{r_+(y)}{r_+(y)+r_-(y)}\in(0,1),\;
q_2(y):=(\|y\|^2-1)^2,
\]
and define
\(
\rho(0,y)=q_1(y)+q_2(y),\;
\rho(1,y)=1-q_1(y)+q_2(y).
\)
Then $\rho$ is continuous and coercive.
However, the infimum of $J_\beta$ over $\pc(\Y)$ is not attained.
\end{example}

\begin{example}[Lower semi-continuity is essential]
\label{ex:measurable-nonlsc-counter}
Let again $\X=\{0,1\}$ with $\mu(0)=\mu(1)=1/2$, and let
\(
\Y=\R
\)
with the usual metric. This is a locally compact Polish space with the Heine-Borel property. Define
\(
q_2(y):=(y^2-1)^2,
\)
and define the Borel function $q_1:\R\to(0,1)$ by
\[
q_1(y)=
\begin{cases}
|y-1|, & 0<|y-1|<\frac14,\\[0.4em]
1-|y+1|, & 0<|y+1|<\frac14,\\[0.4em]
\frac12, & \text{otherwise}.
\end{cases}
\]
Set
\(
\rho(0,y)=q_1(y)+q_2(y),\;
\rho(1,y)=1-q_1(y)+q_2(y).
\)
Then $\rho$ is coercive but not lower semi-continuous. Moreover, $J_\beta$ has a finite infimum that is not attained on $\pc(\Y)$.
\end{example}




\section{Discussion and Future Work}
\label{sec: Discussion}

In this paper, we identify two explicit mechanisms under which optimal reconstructions are attained on the original (compared to the compactified) noncompact reproduction alphabet: a compactification-dominance mechanism for bounded distortions, and a concentration-compactness mechanism for coercive distortions. To the best of our knowledge, this yields the first general original-space existence theory for optimal reconstructions on noncompact Polish alphabets.

The analytical framework presented here opens up several avenues for future research. A primary direction is to address models with additional structural constraints, which are common in applications. For instance, in non-anticipative (or causal) rate-distortion theory \cite{charalambous2014nonanticipative, Charalambous2016Directed, Charalambous2018Nonanticipative, Charalambous2022Complete}, existence results on non-compact spaces with lower semi-continuous distortion remain scarce, largely due to the analytical challenges posed by causal constraints. Adapting a variational approach to rigorously establish the existence of optimal causal coders in such non-compact settings remains a significant open problem for future investigation.
\clearpage
\bibliographystyle{IEEEtran}
\bibliography{rd_noncompact_refs}
\newpage
\onecolumn
\appendices

\section{Proof of Theorem~\ref{thm:bounded-itw}}
\label{app:bounded-itw}

We prove the bounded compactification theorem in the output-law formulation. The case $\beta=0$ is immediate, since $J_0(\nu)=0$ for every $\nu\in\pc(\Y)$. Fix $\beta>0$.

Define the compactified functional on $\pc(\Y^\ast)$ by
\begin{equation}
J_\beta^\ast(\nu):=-\int_\X\log\left(\int_{\Y^\ast}e^{-\beta\rho^\ast(x,y)}\,\nu(\D y)\right)\mu(\D x).
\end{equation}
Since $0\le\rho^\ast\le M$, the inner integral belongs to $[e^{-\beta M},1]$ for every $x$ and every $\nu$. In particular, the logarithm is finite and the integrand is bounded between $0$ and $\beta M$.

We first prove that $J_\beta^\ast$ is weakly lower semi-continuous on $\pc(\Y^\ast)$. Let $\nu_k\rightharpoonup\nu$ weakly in $\pc(\Y^\ast)$. For $\mu$-a.e. $x$, the section $y\mapsto\rho^\ast(x,y)$ is lower semi-continuous on the compact metric space $\Y^\ast$. Hence
\begin{equation}
y\mapsto e^{-\beta\rho^\ast(x,y)}
\end{equation}
is bounded upper semi-continuous. By the Portmanteau theorem,
\begin{equation}
\int_{\Y^\ast}e^{-\beta\rho^\ast(x,y)}\,\nu(\D y)
\ge
\limsup_{k\to\infty}
\int_{\Y^\ast}e^{-\beta\rho^\ast(x,y)}\,\nu_k(\D y).
\end{equation}
Since $z\mapsto -\log z$ is continuous and decreasing on $[e^{-\beta M},1]$,
\begin{equation}
-\log\int_{\Y^\ast}e^{-\beta\rho^\ast(x,y)}\,\nu(\D y)
\le
\liminf_{k\to\infty}
\left[-\log\int_{\Y^\ast}e^{-\beta\rho^\ast(x,y)}\,\nu_k(\D y)\right]
\end{equation}
for $\mu$-a.e. $x$. Fatou's lemma gives
\begin{equation}
J_\beta^\ast(\nu)\le\liminf_{k\to\infty}J_\beta^\ast(\nu_k).
\end{equation}
Thus $J_\beta^\ast$ is weakly lower semi-continuous.

Since $\Y^\ast$ is compact metric, $\mathcal P(\Y^\ast)$ is compact under the topology of weak convergence. Therefore $J_\beta^\ast$ attains its infimum at some $\nu^\ast\in\pc(\Y^\ast)$.

Let $y_0\in\Y$ be given by the compactification-dominance condition and define
\begin{equation}
T:\Y^\ast\to\Y,
\qquad
T(y)=y\quad(y\in\Y),
\qquad
T(\infty)=y_0.
\end{equation}
This map is Borel measurable. Set
\begin{equation}
\widehat\nu:=T_\#\nu^\ast\in\pc(\Y).
\end{equation}
For $\mu$-a.e. $x$ and every $y\in\Y^\ast$, the extension property and compactification-dominance give
\begin{equation}
\rho(x,T(y))\le \rho^\ast(x,y).
\end{equation}
Indeed, if $y\in\Y$, both sides are equal; if $y=\infty$, this is exactly the compactification-dominance condition. Hence
\begin{equation}
e^{-\beta\rho(x,T(y))}\ge e^{-\beta\rho^\ast(x,y)}.
\end{equation}
Consequently, for $\mu$-a.e. $x$,
\begin{equation}
\int_\Y e^{-\beta\rho(x,z)}\,\widehat\nu(\D z)
=
\int_{\Y^\ast} e^{-\beta\rho(x,T(y))}\,\nu^\ast(\D y)
\ge
\int_{\Y^\ast} e^{-\beta\rho^\ast(x,y)}\,\nu^\ast(\D y).
\end{equation}
Taking $-\log$ and integrating yields
\begin{equation}
J_\beta(\widehat\nu)\le J_\beta^\ast(\nu^\ast).
\end{equation}
On the other hand, every $\nu\in\pc(\Y)$ can be regarded as a probability measure on $\Y^\ast$ that assigns zero mass to $\infty$, and for such measures the two functionals agree:
\begin{equation}
J_\beta^\ast(\nu)=J_\beta(\nu).
\end{equation}
Thus
\begin{equation}
\inf_{\pc(\Y^\ast)}J_\beta^\ast\le \inf_{\pc(\Y)}J_\beta.
\end{equation}
Combining this with $\inf_{\pc(\Y)}J_\beta\le J_\beta(\widehat\nu)\le J_\beta^\ast(\nu^\ast)=\inf_{\pc(\Y^\ast)}J_\beta^\ast$ gives
\begin{equation}
J_\beta(\widehat\nu)=\inf_{\nu\in\pc(\Y)}J_\beta(\nu).
\end{equation}
Therefore the infimum of $J_\beta$ is attained on the original alphabet $\Y$.

\section{Proof of Corollary~\ref{cor:zeroone-itw}}
\label{app:zeroone-itw}

Fix $x\in\X=\Y$. Since $\Y$ is Polish, it is Hausdorff; hence the singleton $\{x\}$ is closed. Therefore the set $\Y\setminus\{x\}$ is open, and the section
\begin{equation}
y\mapsto \mathbf 1_{\{x\ne y\}}
\end{equation}
is lower semi-continuous on $\Y$.

It remains to verify lower semi-continuity at the compactification point. Let $y_n\to\infty$ in $\Y^\ast$. By definition of the Alexandroff compactification, the sequence $\{y_n\}$ eventually leaves every compact subset of $\Y$. Since $\{x\}$ is compact, we have $y_n\notin\{x\}$ for all sufficiently large $n$. Thus
\begin{equation}
\rho_0^\ast(x,y_n)=1
\end{equation}
for all sufficiently large $n$, and consequently
\begin{equation}
\liminf_{n\to\infty}\rho_0^\ast(x,y_n)=1=\rho_0^\ast(x,\infty).
\end{equation}
This proves lower semi-continuity of $y\mapsto\rho_0^\ast(x,y)$ on $\Y^\ast$ for every fixed $x$.

The extension is Borel measurable. Indeed,
\begin{equation}
(\rho_0^\ast)^{-1}(\{1\})
=
\{(x,y)\in\Y\times\Y:x\neq y\}
\cup
(\Y\times\{\infty\}).
\end{equation}
The diagonal is closed in $\Y\times\Y$ because $\Y$ is Hausdorff, hence
\begin{equation}
\{(x,y)\in\Y\times\Y:x\neq y\}
\end{equation}
is Borel. Moreover, $\{\infty\}$ is closed in $\Y^\ast$, so $\Y\times\{\infty\}$ is Borel. Therefore $\rho_0^\ast$ is Borel measurable. Finally, for any $y_0\in\Y$,
\begin{equation}
\rho_0(x,y_0)\le 1=\rho_0^\ast(x,\infty),
\end{equation}
so the compactification-dominance condition holds. The conclusion follows from Theorem~\ref{thm:bounded-itw}.

\section{Proof of Theorem~\ref{thm:gcc-itw}}
\label{app:gcc-itw}

The proof follows the classical strategy of \cite{lions1984concentration1, lions1984concentration2,struwe2000variational}, adapted to the abstract Polish-space setting.
We begin by introducing the notion of the concentration function associated with a measure $\mu$:
\begin{equation}
    Q_{\nu}(R) \coloneqq \sup_{y\in \Y} \nu (B_d(y,R)).
\end{equation}

Let $Q_n$ denote the concentration function associated with the measure $\nu_n$.
Observe that $\{Q_n\}$ is a sequence of non-decreasing, nonnegative, and bounded functions on $\mathbb{R}_+$.
Moreover, since each $\nu_n$ is a probability measure, it holds that
\begin{equation}
\lim_{R\to\infty} Q_n(R)=1.
\end{equation}
Complementing this concentration-compactness lemma, our analysis will also leverage a classical result from real analysis, which we state below for completeness.
\begin{lemma}{}[Helly's Selection Theorem]
Let $\{f_n\}$ be a sequence of nondecreasing real-valued functions on an interval $I\subset\R$.
Assume that the sequence is uniformly bounded, i.e., there exist constants $a,b\in\R$ such that $a\le f_n \le b, \forall n$.
Then there exists a subsequence that converges pointwise on $I$.
\end{lemma}

By Helly's selection theorem, there exists a subsequence $\{Q_n\}$ (by a slight abuse of notation) that converges pointwise for almost every $R>0$ to a non-decreasing, nonnegative, bounded function $Q$.
Since any monotone function admits left and right limits and has at most countably many discontinuities, we normalize $Q$ to be left-continuous by defining
\begin{equation}
Q(R)=\lim_{\epsilon\to 0^+} Q(R-\epsilon).
\end{equation}

Fix an arbitrary $R>0$.
We consider a sequence $\{R_k\}$ such that $R_k \nearrow R$ and
\begin{equation}
\lim_{n\to\infty} Q_n(R_k)=Q(R_k).
\end{equation}
Since each $Q_n$ is non-decreasing, we have $Q_n(R_k)\leq Q_n(R)$ for all $k,n\in\mathbb{N}$.
Consequently,
\begin{equation}
\liminf_{n\to \infty} Q_n(R_k)
=
\lim_{n\to \infty} Q_n(R_k)
=
Q(R_k)
\leq
\liminf_{n\to\infty} Q_n(R).
\end{equation}
Letting $R_k\to R$ and using the left-continuity of $Q$, we obtain
\begin{equation}
\lim_{k\to\infty} Q(R_k)=Q(R)\leq \liminf_{n\to\infty} Q_n(R).
\end{equation}

Define
\begin{equation}
\lambda \coloneqq \lim_{R\to\infty} Q(R),
\end{equation}
which clearly satisfies $\lambda\in[0,1]$.
We now distinguish three cases.

\medskip
\noindent
\textbf{Case 1: $\lambda=0$ (Vanishing).}
In this case, by definition of $\lambda$, we have
\begin{equation}
\lim_{n\to\infty}
\left(
\sup_{y\in\Y} \nu_n(B_d(y,R))
\right)=0,
\qquad \forall R>0,
\end{equation}
which is precisely the vanishing alternative.

\medskip
\noindent
\textbf{Case 2: $\lambda=1$ (Compactness).}
Then there exists some $R_0>0$ such that $Q(R_0)>\tfrac12$.
For each $n\in\mathbb{N}$, choose $y_n\in\Y$ satisfying
\begin{equation}
Q_n(R_0)
\leq
\nu_n(B_d(y_n,R_0)) + \frac{1}{n},
\qquad
\lim_{n\to\infty} Q_n(R_0)=Q(R_0).
\end{equation}

Fix $0<\epsilon<\tfrac12$.
Choose $R>0$ such that $Q(R)>1-\epsilon>\tfrac12$.
For each $n$, let $z_n\in\Y$ satisfy
\begin{equation}
Q_n(R)
\leq
\nu_n(B_d(z_n,R)) + \frac{1}{n},
\qquad
\lim_{n\to\infty} Q_n(R)=Q(R),
\end{equation}
where the choice of $z_n$ may depend on $\epsilon$.
Then
\begin{IEEEeqnarray}{rCl}
\nu_n(B_d(y_n,R_0))
+
\nu_n(B_d(z_n,R))
&\ge&
Q_n(R_0)+Q_n(R)-\frac{2}{n} \nonumber\\
&\to&
Q(R_0)+Q(R) \nonumber\\
&>&
\frac12+\frac12=1=\int_{\Y}\,\mathrm{d}\nu_n,
\end{IEEEeqnarray}
for all sufficiently large $n$.
Hence the balls $B_d(y_n,R_0)$ and $B_d(z_n,R)$ must intersect.
Otherwise,
\begin{IEEEeqnarray}{rCl}
\nu_n(B_d(y_n,R_0))
+
\nu_n(B_d(z_n,R))
&=&
\nu_n(B_d(y_n,R_0)\cup B_d(z_n,R)) \nonumber\\
&\le&
\int_{\Y}\,\mathrm{d}\nu_n=1,
\end{IEEEeqnarray}
which leads to a contradiction.

Therefore,
\begin{IEEEeqnarray}{rCl}
&&
\dist(y_n,z_n)\le R_0+R,
\qquad
B_d(z_n, R)\subset B_{d}(y_n, 2R+R_0), \nonumber\\
&&
\nu_n(B_d(y_n, 2R+R_0))
\ge
\nu_n(B_d(z_n,R))
\ge
Q_n(R)-\frac{1}{n}
\ge
1-\epsilon,
\end{IEEEeqnarray}
for all sufficiently large $n$.
Thus there exists $N\in\mathbb{N}$ such that the above holds for all $n\ge N$.
Moreover, since finitely many indices are involved, there exists $\widetilde R>0$ such that
\begin{equation}
\nu_n(B_d(y_n, \widetilde R))\ge 1-\epsilon,
\qquad
\forall\,1\le n\le N.
\end{equation}
Setting $R'=\max\{\widetilde R,\,2R+R_0\}$ yields the compactness alternative.

\medskip
\noindent
\textbf{Case 3: $0<\lambda<1$ (Dichotomy).}
Fix $\epsilon>0$.
There exists $R(\epsilon)>0$ such that
\begin{equation}
Q(R)>\lambda-\frac{\epsilon}{3},
\qquad
\lim_{n\to\infty} Q_n(R)=Q(R).
\end{equation}
Hence there exists $M(\epsilon)\in\mathbb{N}$ such that
\begin{equation}
Q_n(R)>Q(R)-\frac{\epsilon}{3}>\lambda-\frac{2\epsilon}{3},
\qquad
\forall n\ge M(\epsilon).
\end{equation}
Choose a sequence $\{y_n\}$ depending on $\epsilon$ such that
\begin{equation}
Q_n(R)
\ge
\nu_n(B_d(y_n,R))
\ge
Q_n(R)-\frac{\epsilon}{3}
>
\lambda-\epsilon.
\end{equation}

Similarly, there exist $R'(\epsilon)>0$ and $M'(\epsilon)\ge M(\epsilon)$ such that
\begin{equation}
\lambda-\epsilon<Q_n(R')<\lambda+\epsilon,
\qquad
\forall n\ge M'(\epsilon).
\end{equation}
Since $\lim_{r\to\infty} Q(r)=\lambda$, there exists $R_0<\infty$ such that
\begin{equation}
Q(r)<\lambda+\frac{\epsilon}{2},
\qquad
\forall r\ge R_0.
\end{equation}
Choose a strictly increasing sequence $\{R_k\}_{k\ge1}$ with $R_k\to\infty$ and $R_1\ge R_0$, such that
\begin{equation}
Q(R_k)<\lambda+\frac{\epsilon}{2},
\qquad
\lim_{n\to\infty} Q_n(R_k)=Q(R_k),
\quad \forall k\in\mathbb{N}.
\end{equation}

By a diagonal construction, we define an increasing index sequence $\{M_k\}_{k\ge1}$ such that
\begin{equation}
Q_n(R_k)<Q(R_k)+\frac{\epsilon}{2}<\lambda+\epsilon,
\qquad
\forall n\ge M_k.
\end{equation}
Using $\{R_k\}$ and $\{M_k\}$, define a sequence $\{\widetilde R_n\}$ by
\begin{equation}
\widetilde R_n=
\begin{cases}
R_0, & M'\le n<M_1,\\
R_k, & M_k\le n<M_{k+1}.
\end{cases}
\end{equation}

In summary, for any $\epsilon>0$, there exist $R>0$, $n_0\in\mathbb{N}$, a sequence $\{R_n\}$ with $R_n\to\infty$, and points $\{y_n\}\subset\Y$ such that for all $n\ge n_0$,
\begin{IEEEeqnarray}{rCl}
&&
Q_n(R)\ge \nu_n(B_d(y_n,R))>\lambda-\epsilon,\\
&&
Q_n(R)\le Q_n(R_n)<\lambda+\epsilon.
\end{IEEEeqnarray}
For any $R'>R$, we may increase $n_0$ so that $R_n\ge R'$ for all $n\ge n_0$.
Define
\begin{equation}
\nu_n^1=\nu_n\cdot\bm{1}_{B_d(y_n,R)},
\qquad
\nu_n^2=\nu_n\cdot\bm{1}_{\Y\setminus B_d(y_n,R_n)}.
\end{equation}
Then $0\le \nu_n^1+\nu_n^2\le \nu_n$, with
\begin{equation}
\supp(\nu_n^1)\subset B_d(y_n,R),
\qquad
\supp(\nu_n^2)\subset \Y\setminus B_d(y_n,R_n)\subset \Y\setminus B_d(y_n,R').
\end{equation}
Finally, for all $n\ge n_0$,
\begin{IEEEeqnarray}{rCl}
\left|
\lambda-\nu_n^1(\Y)
\right|
+
\left|
(1-\lambda)-\nu_n^2(\Y)
\right| =
\left|
\lambda-\nu_n(B_d(y_n,R))
\right|
+
\left|
\lambda-\nu_n(B_d(y_n,R_n))
\right| < 
2\epsilon.
\end{IEEEeqnarray}
This establishes the dichotomy alternative and completes the proof.

\section{Proof of Theorem~\ref{thm:coercive-itw}}
\label{app:coercive-itw}

Let
\(
m_\beta:=\inf_{\nu\in\pc(\Y)}J_\beta(\nu)<\infty.
\)
Throughout the proof, fix a compatible Heine--Borel metric $d$ on $\Y$.
Let $\{\nu_n\}\subset\pc(\Y)$ be a minimizing sequence:
\begin{equation}
J_\beta(\nu_n)\to m_\beta.
\end{equation}
By passing to a subsequence if necessary, we may always suppose that
\begin{equation}
J_\beta(\nu_n)\le m_\beta+\frac1n.
\end{equation}
The concentration-compactness alternative is subsequential. Thus, in Steps~1 and~2 below, we prove that no subsequence of a minimizing sequence can fall into the vanishing or dichotomy alternatives.

For notational convenience, define
\begin{equation}
\Phi_n(x)
:=
\int_\Y e^{-\beta\rho(x,y)}\,\nu_n(\D y).
\end{equation}
Then
\begin{equation}
J_\beta(\nu_n)
=
-\int_\X \log \Phi_n(x)\,\mu(\D x).
\end{equation}

\medskip
\noindent
\textbf{Step 1: Ruling out Vanishing.}
Suppose that a subsequence of $\{\nu_n\}$, still denoted by $\{\nu_n\}$, satisfies the vanishing alternative:
\begin{equation}
\sup_{y\in\Y}\nu_n(B_d(y,R))\to0
\qquad\text{for every }R>0.
\end{equation}
Then for every compact set $K\subset\Y$,
\begin{equation}
\nu_n(K)\to0.
\end{equation}
Indeed, fix $R>0$. Since $K$ is compact, there exist finitely many points $y_1,\ldots,y_N\in\Y$ such that
\begin{equation}
K\subset \bigcup_{i=1}^N B_d(y_i,R).
\end{equation}
Therefore
\begin{equation}
\nu_n(K)
\le
\sum_{i=1}^N \nu_n(B_d(y_i,R))
\le
N\sup_{y\in\Y}\nu_n(B_d(y,R))
\to0.
\end{equation}

For $\mu$-a.e. $x$ and every $L<\infty$, compact-sublevel coercivity gives a compact set
\begin{equation}
K_{x,L}:=\{y\in\Y:\rho(x,y)\le L\}.
\end{equation}
Splitting the integral defining $\Phi_n(x)$ over $K_{x,L}$ and its complement gives
\begin{equation}
\Phi_n(x)
=
\int_{K_{x,L}} e^{-\beta\rho(x,y)}\,\nu_n(\D y)
+
\int_{\Y\setminus K_{x,L}} e^{-\beta\rho(x,y)}\,\nu_n(\D y).
\end{equation}
Since the first integrand is bounded by $1$, and on $\Y\setminus K_{x,L}$ one has $\rho(x,y)>L$, we obtain
\begin{equation}
\Phi_n(x)
\le
\nu_n(K_{x,L})+e^{-\beta L}.
\end{equation}
Letting $n\to\infty$ and using $\nu_n(K_{x,L})\to0$, we get
\begin{equation}
\limsup_{n\to\infty}\Phi_n(x)\le e^{-\beta L}.
\end{equation}
Since $L<\infty$ is arbitrary,
\begin{equation}
\Phi_n(x)\to0
\qquad
\text{for }\mu\text{-a.e. }x.
\end{equation}
Thus
\begin{equation}
-\log\Phi_n(x)\to+\infty
\qquad
\text{for }\mu\text{-a.e. }x.
\end{equation}
By Fatou's lemma,
\begin{equation}
\liminf_{n\to\infty}J_\beta(\nu_n)
=
\liminf_{n\to\infty}
\int_\X -\log\Phi_n(x)\,\mu(\D x)
=
+\infty,
\end{equation}
which contradicts
\begin{equation}
J_\beta(\nu_n)\to m_\beta<\infty.
\end{equation}
Hence no minimizing subsequence can satisfy vanishing.

\medskip
\noindent
\textbf{Step 2: Ruling out Dichotomy.}
Suppose that a subsequence of $\{\nu_n\}$, still denoted by $\{\nu_n\}$, satisfies the dichotomy alternative. 
Thus there exists $\alpha\in(0,1)$ such that, for every sufficiently small $\varepsilon>0$, one can find the dichotomy decomposition
\begin{equation}
0\le \nu_n^1+\nu_n^2\le \nu_n,
\qquad
\nu_n^3:=\nu_n-\nu_n^1-\nu_n^2\ge0,
\end{equation}
with
\begin{equation}
\nu_n^1(\Y)\to\alpha,
\qquad
\nu_n^2(\Y)\to1-\alpha
\end{equation}
up to an asymptotic error of order $\varepsilon$, and
\begin{equation}
\limsup_{n\to\infty}\nu_n^3(\Y)\le C\varepsilon,
\end{equation}
where $C$ is an absolute constant. Most importantly,
\begin{equation}
\operatorname{dist}
\bigl(
\supp(\nu_n^1),\supp(\nu_n^2)
\bigr)\to\infty.
\tag{Sep}
\end{equation}

The support separation in (Sep) is the key geometric input. In the dichotomy alternative, one principal component is localized in a fixed-radius ball, while the other is placed outside larger and larger balls with the same center. If those centers escape to infinity, then the localized component escapes every compact subset of $\Y$. If those centers do not escape, then, by the Heine--Borel property, the fixed-radius balls remain inside a common compact set along a further subsequence, and the separation condition (Sep) forces the other principal component to escape every compact subset of $\Y$. Hence at least one principal component escapes every compact subset of $\Y$.

If both principal components escaped every compact subset of $\Y$, then for $\mu$-a.e. $x$ and every $L<\infty$, the compact set $K_{x,L}$ would eventually be disjoint from both principal supports. Hence, for all sufficiently large $n$,
\begin{equation}
\int_\Y e^{-\beta\rho(x,y)}\,\nu_n^1(\D y)\le e^{-\beta L}\nu_n^1(\Y),
\end{equation}
and similarly,
\begin{equation}
\int_\Y e^{-\beta\rho(x,y)}\,\nu_n^2(\D y)\le e^{-\beta L}\nu_n^2(\Y).
\end{equation}
Moreover,
\begin{equation}
\int_\Y e^{-\beta\rho(x,y)}\,\nu_n^3(\D y)
\le \nu_n^3(\Y).
\end{equation}
Thus
\begin{equation}
\limsup_{n\to\infty}\Phi_n(x)
\le
2e^{-\beta L}+C\varepsilon.
\end{equation}
Letting $L\to\infty$ gives
\begin{equation}
\limsup_{n\to\infty}\Phi_n(x)
\le C\varepsilon
\qquad
\mu\text{-a.e. }x.
\end{equation}
Choosing $\varepsilon>0$ so small that
\begin{equation}
-\log(C\varepsilon)>m_\beta+1
\end{equation}
and applying Fatou's lemma gives
\begin{equation}
\liminf_{n\to\infty}J_\beta(\nu_n)\ge -\log(C\varepsilon)>m_\beta+1,
\end{equation}
contradicting $J_\beta(\nu_n)\to m_\beta$. Therefore not both principal components can escape.

Consequently, after relabeling the two principal components if necessary, exactly one of them escapes every compact subset of $\Y$. Denote this escaping component by $\nu_n^{\mathsf{es}}$, and set
\begin{equation}
\kappa_n:=\nu_n^{\mathsf{es}}(\Y).
\end{equation}
Since $\nu_n^{\mathsf{es}}$ is a principal dichotomy component, for $\varepsilon>0$ sufficiently small there exist constants
\begin{equation}
0<\kappa_0<\kappa_1<1
\end{equation}
such that
\begin{equation}
\kappa_n\in[\kappa_0,\kappa_1]
\end{equation}
for all large $n$.

Define
\begin{equation}
\Phi_n^{\mathsf{es}}(x)
:=
\int_\Y e^{-\beta\rho(x,y)}\,\nu_n^{\mathsf{es}}(\D y).
\end{equation}
Since $\nu_n^{\mathsf{es}}$ escapes every compact set, for $\mu$-a.e. $x$ and every $L<\infty$, the compact set $K_{x,L}$ is eventually disjoint from $\supp(\nu_n^{\mathsf{es}})$. Therefore, for all sufficiently large $n$,
\begin{equation}
\Phi_n^{\mathsf{es}}(x)
\le e^{-\beta L}\nu_n^{\mathsf{es}}(\Y)
\le e^{-\beta L}.
\end{equation}
Letting $L\to\infty$ gives
\begin{equation}
\Phi_n^{\mathsf{es}}(x)\to0
\qquad
\mu\text{-a.e. }x.
\end{equation}
Since $0\le \Phi_n^{\mathsf{es}}(x)\le1$, dominated convergence yields
\begin{equation}\label{eq:es-vanish}
\int_\X \Phi_n^{\mathsf{es}}(x)\,\mu(\D x)\to0.
\end{equation}

Next define the relative contribution of the escaping component by
\begin{equation}
r_n(x):=
\begin{cases}
\dfrac{\Phi_n^{\mathsf{es}}(x)}{\Phi_n(x)}, & \Phi_n(x)>0,\\[1em]
0, & \Phi_n(x)=0.
\end{cases}
\end{equation}
Since $\nu_n^{\mathsf{es}}\le\nu_n$, we have
\begin{equation}
0\le r_n(x)\le1.
\end{equation}
We claim that
\begin{equation}\label{eq:rn-vanish}
r_n\to0
\qquad\text{in }\mu\text{-measure}.
\end{equation}
Because the sequence is minimizing and $m_\beta<\infty$, there exists $C_0<\infty$ such that
\begin{equation}
J_\beta(\nu_n)
=
\int_\X -\log\Phi_n(x)\,\mu(\D x)
\le C_0
\end{equation}
for all large $n$. Fix $\gamma>0$ and $\delta>0$. For any $a\in(0,1)$,
\begin{equation}
\{r_n>\gamma\}
\subset
\{\Phi_n<a\}\cup\{\Phi_n^{\mathsf{es}}>\gamma a\}.
\end{equation}
Indeed, if $r_n(x)>\gamma$ and $\Phi_n(x)\ge a$, then
\begin{equation}
\Phi_n^{\mathsf{es}}(x)=r_n(x)\Phi_n(x)>\gamma a.
\end{equation}
The first set is controlled by Markov's inequality:
\begin{equation}
\mu(\Phi_n<a)
=
\mu(-\log\Phi_n>-\log a)
\le
\frac{C_0}{-\log a}.
\end{equation}
Choose $a>0$ so small that
\begin{equation}
\frac{C_0}{-\log a}<\frac{\delta}{2}.
\end{equation}
For this fixed $a$, \eqref{eq:es-vanish} gives
\begin{equation}
\mu(\Phi_n^{\mathsf{es}}>\gamma a)
\le
\frac{1}{\gamma a}
\int_\X \Phi_n^{\mathsf{es}}(x)\,\mu(\D x)
\to0.
\end{equation}
Hence $\mu(r_n>\gamma)<\delta$ for all sufficiently large $n$, proving \eqref{eq:rn-vanish}.

Now construct a competitor by moving a small fraction of the escaping mass to the complementary part of $\nu_n$. 
For a fixed $\tau\in(0,1)$, define
\begin{equation}
\widetilde\nu_n
:=
(1-\tau)\nu_n^{\mathsf{es}}
+
\left(1+\frac{\tau\kappa_n}{1-\kappa_n}\right)
(\nu_n-\nu_n^{\mathsf{es}}).
\end{equation}
This is a probability measure because
\begin{equation}
(1-\tau)\kappa_n
+
\left(1+\frac{\tau\kappa_n}{1-\kappa_n}\right)
(1-\kappa_n)
=
1.
\end{equation}
Let
\begin{equation}
\widetilde\Phi_n(x)
:=
\int_\Y e^{-\beta\rho(x,y)}\,\widetilde\nu_n(\D y).
\end{equation}
A direct calculation gives
\begin{equation}
\frac{\widetilde\Phi_n(x)}{\Phi_n(x)}
=
1+\tau\,\frac{\kappa_n-r_n(x)}{1-\kappa_n}
\end{equation}
whenever $\Phi_n(x)>0$. Since $0\le r_n\le1$, we have the global lower bound
\begin{equation}
\frac{\widetilde\Phi_n(x)}{\Phi_n(x)}\ge1-\tau.
\end{equation}
On the set
\begin{equation}
E_n:=\{x:r_n(x)\le\kappa_0/2\},
\end{equation}
we have
\begin{equation}
\frac{\widetilde\Phi_n(x)}{\Phi_n(x)}
\ge
1+\frac{\tau\kappa_0}{2}.
\end{equation}
By \eqref{eq:rn-vanish}, $\mu(E_n)\to1$.

Therefore
\begin{equation}
\begin{aligned}
J_\beta(\nu_n)-J_\beta(\widetilde\nu_n)
&=
\int_\X
\log\frac{\widetilde\Phi_n(x)}{\Phi_n(x)}
\,\mu(\D x)\\
&\ge
\mu(E_n)\log\left(1+\frac{\tau\kappa_0}{2}\right)
+
(1-\mu(E_n))\log(1-\tau).
\end{aligned}
\end{equation}
Since $\mu(E_n)\to1$, the right-hand side is bounded below by a positive constant along a tail. Thus, for all sufficiently large $n$,
\begin{equation}
J_\beta(\widetilde\nu_n)<J_\beta(\nu_n)
\end{equation}
with a uniform improvement. This contradicts the minimizing property
\begin{equation}
J_\beta(\nu_n)\to m_\beta.
\end{equation}
Hence dichotomy cannot occur.

\medskip
\noindent
\textbf{Step 3: Compactness and Existence.}
Let $\{\nu_n\}$ be any minimizing sequence. 
By the concentration-compactness alternative, after passing to a subsequence, one of compactness, vanishing, or dichotomy occurs. 
Steps~1 and~2 rule out vanishing and dichotomy. Hence the compactness alternative must occur along a subsequence.

Thus, for every $\eta>0$, there exist $R_\eta<\infty$ and centers $y_n^\eta\in\Y$ such that
\begin{equation}
\nu_n(B_d(y_n^\eta,R_\eta))\ge1-\eta
\end{equation}
along the subsequence.

We first show that, for each fixed $\eta>0$ small enough, the centers $y_n^\eta$ cannot escape every compact subset of $\Y$. Suppose they did. Then the balls $B_d(y_n^\eta,R_\eta)$ would eventually be disjoint from every compact set. For $\mu$-a.e. $x$ and every $L<\infty$, the compact set $K_{x,L}$ would eventually be disjoint from these balls. Therefore
\begin{equation}
\Phi_n(x)
\le
e^{-\beta L}+\eta.
\end{equation}
Letting $L\to\infty$ gives
\begin{equation}
\limsup_{n\to\infty}\Phi_n(x)\le\eta
\qquad
\mu\text{-a.e. }x.
\end{equation}
If $\eta>0$ is chosen so small that
\begin{equation}
-\log\eta>m_\beta+1,
\end{equation}
then Fatou's lemma gives
\begin{equation}
\liminf_{n\to\infty}J_\beta(\nu_n)>m_\beta,
\end{equation}
contradicting $J_\beta(\nu_n)\to m_\beta$.

Therefore the concentration centers do not escape. Since the metric has the Heine--Borel property, for each fixed $\eta$ we can pass to a further subsequence along which the corresponding centers remain in a common compact set. Consequently, the balls \(B_d(y_n^\eta,R_\eta)\) are contained in a common compact set. A standard diagonal argument over \(\eta=1/2^k\) gives a tight subsequence of \(\{\nu_n\}\). Specifically, let
\begin{equation}
\eta_k:=2^{-k},\qquad k\ge1.
\end{equation}
Start from the subsequence on which the compactness alternative holds, and denote its index set by \(S_0\).

For \(k=1\), apply the compactness alternative with \(\eta_1\). There exist \(R_1<\infty\) and centers \(y_n^{(1)}\), \(n\in S_0\), such that
\begin{equation}
\nu_n\bigl(B_d(y_n^{(1)},R_1)\bigr)\ge 1-\eta_1,
\qquad n\in S_0.
\end{equation}
By the previous argument, these centers cannot escape to infinity. Hence there exists an infinite subset
\begin{equation}
S_1\subset S_0
\end{equation}
and a compact set \(C_1\subset\Y\) such that
\begin{equation}
y_n^{(1)}\in C_1,\qquad n\in S_1.
\end{equation}
Define
\begin{equation}
K_1:=\{y\in\Y:d(y,C_1)\le R_1\}.
\end{equation}
Since \(C_1\) is compact and the metric has the Heine--Borel property, \(K_1\) is compact. Moreover,
\begin{equation}
B_d(y_n^{(1)},R_1)\subset K_1,\qquad n\in S_1,
\end{equation}
and therefore
\begin{equation}
\nu_n(K_1)\ge 1-\eta_1,\qquad n\in S_1.
\end{equation}

Next, suppose that \(S_{k-1}\) has been constructed. Apply the compactness alternative to the subsequence indexed by \(S_{k-1}\) with \(\eta_k\). We obtain \(R_k<\infty\) and centers \(y_n^{(k)}\), \(n\in S_{k-1}\), such that
\begin{equation}
\nu_n\bigl(B_d(y_n^{(k)},R_k)\bigr)\ge 1-\eta_k,
\qquad n\in S_{k-1}.
\end{equation}
Again, the centers cannot escape to infinity. Hence there exists an infinite subset
\begin{equation}
S_k\subset S_{k-1}
\end{equation}
and a compact set \(C_k\subset\Y\) such that
\begin{equation}
y_n^{(k)}\in C_k,\qquad n\in S_k.
\end{equation}
Set
\begin{equation}
K_k:=\{y\in\Y:d(y,C_k)\le R_k\}.
\end{equation}
Then \(K_k\) is compact, and
\begin{equation}
\nu_n(K_k)\ge 1-\eta_k,\qquad n\in S_k.
\end{equation}

This gives a nested sequence of infinite index sets
\begin{equation}
S_0\supset S_1\supset S_2\supset\cdots,
\end{equation}
and compact sets \(K_k\subset\Y\) such that
\begin{equation}
\nu_n(K_k)\ge 1-\eta_k
\qquad\text{for all }n\in S_k.
\end{equation}

Now choose the diagonal subsequence as follows. Let \(m_j\) be the \(j\)-th element of \(S_j\), chosen increasingly. Then, for each fixed \(k\), whenever \(j\ge k\), we have
\begin{equation}
m_j\in S_j\subset S_k.
\end{equation}
Therefore,
\begin{equation}
\nu_{m_j}(K_k)\ge 1-\eta_k,
\qquad j\ge k.
\end{equation}

We claim that \(\{\nu_{m_j}\}\) is tight. Let \(\varepsilon>0\). Choose \(k\) so large that
\begin{equation}
\eta_k<\frac{\varepsilon}{2}.
\end{equation}
Then
\begin{equation}
\nu_{m_j}(K_k)\ge 1-\frac{\varepsilon}{2},
\qquad j\ge k.
\end{equation}
This controls the tail of the diagonal subsequence. For the finitely many initial terms \(j<k\), each measure \(\nu_{m_j}\) is tight because \(\Y\) is Polish. Hence for each \(j<k\), there exists a compact set \(F_j\subset\Y\) such that
\begin{equation}
\nu_{m_j}(F_j)\ge1-\varepsilon.
\end{equation}
Define
\begin{equation}
K:=K_k\cup F_1\cup\cdots\cup F_{k-1}.
\end{equation}
This is compact as a finite union of compact sets. For \(j\ge k\),
\begin{equation}
\nu_{m_j}(K)\ge \nu_{m_j}(K_k)\ge1-\frac{\varepsilon}{2}\ge1-\varepsilon,
\end{equation}
and for \(j<k\),
\begin{equation}
\nu_{m_j}(K)\ge\nu_{m_j}(F_j)\ge1-\varepsilon.
\end{equation}
Thus
\begin{equation}
\nu_{m_j}(K)\ge1-\varepsilon
\qquad\text{for all }j.
\end{equation}
Since \(\varepsilon>0\) was arbitrary, the diagonal subsequence \(\{\nu_{m_j}\}\) is tight.

By Prokhorov's theorem, this tight subsequence has a weakly convergent further subsequence,
\begin{equation}
\nu_{n_j}\rightharpoonup \nu^\ast\in\pc(\Y).
\end{equation}
By the lower semi-continuity of \(J_\beta\) with respect to weak convergence, we have
\begin{equation}
J_\beta(\nu^\ast)
\le
\liminf_{j\to\infty}J_\beta(\nu_{n_j})
=
m_\beta.
\end{equation}
Since \(m_\beta\) is the infimum, \(J_\beta(\nu^\ast)=m_\beta\). Therefore \(J_\beta\) attains its infimum.

\medskip
\noindent
\textbf{Step 4: Every minimizing sequence is tight.}
We now prove the stronger statement that every minimizing sequence is tight.

Suppose, toward contradiction, that there exists a minimizing sequence \(\{\nu_n\}\) that is not tight. 
Fix \(y_0\in\Y\) and set
\begin{equation}
K_j:=B_d(y_0,j),
\qquad j\in\mathbb N
\end{equation}
to be closed balls. Since \(d\) has the Heine--Borel property, each \(K_j\) is compact. Moreover,
\begin{equation}
K_1\subset K_2\subset\cdots,
\qquad
\bigcup_{j=1}^\infty K_j=\Y.
\end{equation}
The last equality follows because for every $y\in\Y$, the distance $d(y,y_0)$ is finite, so $y\in K_j$ for every integer $j>d(y,y_0)$.

We also note that every compact subset of $\Y$ is contained in some $K_j$. Since $\{\nu_n\}$ is not tight, there exists $\varepsilon_0>0$ such that no compact set captures mass $1-\varepsilon_0$ uniformly in $n$. Equivalently,
\begin{equation}\label{eq:non-tight}
\forall K\Subset\Y,\qquad
\exists n\in\mathbb N
\quad\text{such that}\quad
\nu_n(K)<1-\varepsilon_0.
\end{equation}
In fact, this failure occurs arbitrarily far along the sequence: for every compact $K\subset\Y$ and every $N\in\mathbb N$, there exists $n\ge N$ such that
\begin{equation}\label{eq:non-tight-strong}
\nu_n(K)<1-\varepsilon_0.
\end{equation}
To see this, suppose the contrary. Then there would exist a compact set $K\subset\Y$ and an integer $N$ such that
\begin{equation}
\nu_n(K)\ge 1-\varepsilon_0
\qquad\text{for all }n\ge N.
\end{equation}
For each of the finitely many indices $1\le n<N$, since $K_j\uparrow\Y$, the continuity of measures from below gives
\begin{equation}
\nu_n(K_j)\uparrow \nu_n(\Y)=1.
\end{equation}
Hence we may choose $j_n$ such that
\begin{equation}
\nu_n(K_{j_n})\ge 1-\varepsilon_0,
\qquad 1\le n<N.
\end{equation}
Now define the compact set
\begin{equation}
K':=K\cup\bigcup_{1\le n<N}K_{j_n}.
\end{equation}
This is compact as a finite union of compact sets. For $n\ge N$,
\begin{equation}
\nu_n(K')\ge \nu_n(K)\ge1-\varepsilon_0,
\end{equation}
while for $1\le n<N$,
\begin{equation}
\nu_n(K')\ge\nu_n(K_{j_n})\ge1-\varepsilon_0.
\end{equation}
Thus
\begin{equation}
\nu_n(K')\ge1-\varepsilon_0
\qquad\text{for all }n,
\end{equation}
which contradicts the choice of $\varepsilon_0$ in \eqref{eq:non-tight}. Therefore \eqref{eq:non-tight-strong} holds.

Using \eqref{eq:non-tight-strong} inductively with $K=K_j$, we can choose a strictly increasing subsequence $\{n_j\}_{j\ge1}$ such that
\begin{equation}
\nu_{n_j}(K_j)<1-\varepsilon_0,
\qquad j\ge1.
\end{equation}
Relabel this subsequence as $\{\nu_j\}$, i.e.,
\begin{equation}
\nu_j:=\nu_{n_j}.
\end{equation}
Then
\begin{equation}\label{eq: non-tight-subseq}
\nu_j(K_j)<1-\varepsilon_0,
\qquad j\ge1.
\end{equation}

This relabeled subsequence has no tight further subsequence. Indeed, suppose that a further subsequence $\{\nu_{j_\ell}\}$ were tight. Then there would exist a compact set $K\subset\Y$ such that
\begin{equation}
\nu_{j_\ell}(K)\ge1-\frac{\varepsilon_0}{2}
\qquad\text{for all sufficiently large }\ell.
\end{equation}
As shown above, since $K$ is compact, there exists $J<\infty$ such that
\begin{equation}
K\subset K_J.
\end{equation}
For all sufficiently large $\ell$, we have $j_\ell\ge J$, and hence
\begin{equation}
K\subset K_J\subset K_{j_\ell}.
\end{equation}
Therefore,
\begin{equation}
\nu_{j_\ell}(K_{j_\ell})
\ge
\nu_{j_\ell}(K)
\ge
1-\frac{\varepsilon_0}{2}
\end{equation}
for all sufficiently large $\ell$, contradicting \eqref{eq: non-tight-subseq}, since
\begin{equation}
1-\frac{\varepsilon_0}{2}>1-\varepsilon_0.
\end{equation}
Thus the bad subsequence has no tight further subsequence.

On the other hand, Step~3 applies to every minimizing sequence and therefore to this bad subsequence. Hence it must have a tight further subsequence. This contradiction proves that the original minimizing sequence must be tight.

Finally, if \(\bar\nu\) is any weak cluster point of any minimizing sequence, then the lower semi-continuity of \(J_\beta\) with respect to weak convergence gives
\begin{equation}
J_\beta(\bar\nu)
\le
\liminf_{n\to\infty}J_\beta(\nu_n)
=
m_\beta.
\end{equation}
Hence \(J_\beta(\bar\nu)=m_\beta\), so every weak cluster point of a minimizing sequence is a minimizer.

\section{A Uniform Nonattainment Scheme and Counterexamples}
\label{app:counterexamples}

This appendix proves the counterexamples stated in the main text. The first example shows that coercivity cannot be omitted directly. The next two examples use a uniform construction to show that mere Polishness without local compactness and measurability without lower semi-continuity are insufficient.

\subsection{Coercivity is essential}

Let $\X=\{\ast\}$, $\Y=\R$, and $\rho(\ast,y)=(1+y^2)^{-1}$. For every finite $y$, one has $e^{-\beta\rho(\ast,y)}<1$. Hence for any $\nu\in\pc(\R)$,
\begin{equation}
\int_\R e^{-\beta\rho(\ast,y)}\,\nu(\D y)<1,
\end{equation}
and therefore $J_\beta(\nu)>0$. On the other hand,
\begin{equation}
J_\beta(\delta_n)=-\log e^{-\beta/(1+n^2)}=\frac{\beta}{1+n^2}\to0.
\end{equation}
Thus the infimum is zero but is not attained.

\subsection{The phantom-endpoint construction}

Let $\X=\{0,1\}$ with $\mu(0)=\mu(1)=1/2$. Let $\Y$ be a measurable reproduction space and suppose that there exist measurable functions
\begin{equation}
q_1:\Y\to(0,1),\qquad q_2:\Y\to[0,\infty)
\end{equation}
such that
\begin{equation}
\rho(0,y)=q_1(y)+q_2(y),
\qquad
\rho(1,y)=1-q_1(y)+q_2(y).
\end{equation}
Assume that there exist sequences $y_n^0,y_n^1\in\Y$ with
\begin{equation}
q_1(y_n^0)\to0,
\quad q_2(y_n^0)\to0,
\qquad
q_1(y_n^1)\to1,
\quad q_2(y_n^1)\to0.
\end{equation}
We show that $J_\beta$ has a finite infimum that is not attained for any $\beta>0$.

Set $c=e^{-\beta}$. For any $\nu\in\pc(\Y)$ define
\begin{equation}
A_\nu:=\int_\Y e^{-\beta(q_1(y)+q_2(y))}\,\nu(\D y),
\qquad
B_\nu:=\int_\Y e^{-\beta(1-q_1(y)+q_2(y))}\,\nu(\D y).
\end{equation}
Then
\begin{equation}
J_\beta(\nu)=-\frac12\log A_\nu-\frac12\log B_\nu.
\end{equation}
For every $y\in\Y$, since $q_1(y)\in(0,1)$ and $q_2(y)\ge0$,
\begin{equation}
e^{-\beta(q_1(y)+q_2(y))}+e^{-\beta(1-q_1(y)+q_2(y))}=e^{-\beta q_2(y)}(e^{-\beta q_1 (y)}+e^{-\beta (1-q_1(y))})<1\cdot (1+c)=1+c.
\end{equation}
Therefore $A_\nu+B_\nu<1+c$, and by the arithmetic-geometric mean inequality,
\begin{equation}
A_\nu B_\nu<\left(\frac{1+c}{2}\right)^2.
\end{equation}
Thus
\begin{equation}
J_\beta(\nu)>-\log\left(\frac{1+e^{-\beta}}{2}\right)
\end{equation}
for every $\nu\in\pc(\Y)$.

On the other hand, take
\begin{equation}
\nu_n:=\frac12\delta_{y_n^0}+\frac12\delta_{y_n^1}.
\end{equation}
Then $A_{\nu_n}\to(1+c)/2$ and $B_{\nu_n}\to(1+c)/2$, hence
\begin{equation}
J_\beta(\nu_n)\to -\log\left(\frac{1+e^{-\beta}}{2}\right).
\end{equation}
Thus the infimum equals this value but is not attained. The missing optimizers are the phantom endpoints corresponding to $(q_1,q_2)=(0,0)$ and $(q_1,q_2)=(1,0)$.

\subsection{Polishness without local compactness is insufficient}

Let $\Y=\ell^2$ with its Hilbert norm, and let $\{e_n\}$ be the standard orthonormal basis. Define
\begin{equation}
r_+(y):=\inf_{n\ge1}\left(\frac1n+\|y-e_n\|\right),
\qquad
r_-(y):=\inf_{n\ge1}\left(\frac1n+\|y+e_n\|\right).
\end{equation}
Both functions are $1$-Lipschitz. For each fixed $y\in\ell^2$,
\begin{equation}
\|y-e_n\|^2=\|y\|^2+1-2\langle y,e_n\rangle\to \|y\|^2+1>0,
\end{equation}
and similarly $\|y+e_n\|^2\to\|y\|^2+1$. Hence $r_+(y)>0$ and $r_-(y)>0$ for every fixed $y$.

Set
\begin{equation}
q_1(y):=\frac{r_+(y)}{r_+(y)+r_-(y)},
\qquad
q_2(y):=(\|y\|^2-1)^2.
\end{equation}
Then $q_1$ and $q_2$ are continuous, $q_1(y)\in(0,1)$, and the resulting distortion is continuous. Moreover $q_2(y)\to\infty$ as $\|y\|\to\infty$, so the distortion has metric growth at infinity. However,
\begin{equation}
q_2(e_n)=q_2(-e_n)=0,
\qquad
q_1(e_n)\to0,
\qquad
q_1(-e_n)\to1.
\end{equation}
The phantom-endpoint construction therefore applies with $y_n^0=e_n$ and $y_n^1=-e_n$. Hence $J_\beta$ has no minimizer on $\pc(\ell^2)$ for any $\beta>0$. This shows that Polishness and metric growth at infinity do not replace local compactness and compact-sublevel coercivity.

\subsection{Measurability without lower semi-continuity is insufficient}

Let $\Y=\R$ with its usual metric and define
\begin{equation}
q_2(y):=(y^2-1)^2.
\end{equation}
Let $q_1:\R\to(0,1)$ be the Borel function
\begin{equation}
q_1(y)=
\begin{cases}
|y-1|, & 0<|y-1|<\frac14,\\[0.4em]
1-|y+1|, & 0<|y+1|<\frac14,\\[0.4em]
\frac12, & \text{otherwise}.
\end{cases}
\end{equation}
Set
\begin{equation}
\rho(0,y)=q_1(y)+q_2(y),
\qquad
\rho(1,y)=1-q_1(y)+q_2(y).
\end{equation}
Then $\rho$ is Borel measurable and $q_2(y)\to\infty$ as $|y|\to\infty$, so the distortion has metric growth at infinity. It is not lower semi-continuous: $q_1(1)=1/2$ whereas $\liminf_{y\to1}q_1(y)=0$, hence
\begin{equation}
\rho(0,1)=\frac12,
\qquad
\liminf_{y\to1}\rho(0,y)=0.
\end{equation}
Taking
\begin{equation}
y_n^0=1+\frac1n,
\qquad
 y_n^1=-1+\frac1n,
\end{equation}
we have
\begin{equation}
q_2(y_n^0)\to0,
\quad q_1(y_n^0)\to0,
\qquad
 q_2(y_n^1)\to0,
\quad q_1(y_n^1)\to1.
\end{equation}
Thus the phantom-endpoint construction applies again, and $J_\beta$ has no minimizer on $\pc(\R)$ for any $\beta>0$.

\end{document}